\documentclass[11pt]{amsart}
\usepackage{amsfonts}
\usepackage{amscd}
\usepackage{mathrsfs}

\newtheorem{theorem}{Theorem}[section]

\newtheorem{corollary}[theorem]{Corollary}
\theoremstyle{definition}
\newtheorem{definition}[theorem]{Definition}

\newtheorem{remark}[theorem]{Remark}
\numberwithin{equation}{section}

\newsymbol\subsetneq 2328
\input{tcilatex}

\begin{document}
\title[Quantum mechanics in phase space...]{Quantum mechanics in phase
space: The Schr\"odinger and the Moyal representations}
\author[N.C. Dias]{Nuno Costa Dias}
\author[M. de Gosson]{Maurice de Gosson}
\author[F. Luef]{Franz Luef}
\author[J.N. Prata]{Jo\~{a}o Nuno Prata}

\begin{abstract}
We present a phase space formulation of quantum mechanics in the
Schr\"odinger representation and derive the associated Weyl
pseudo-differential calculus. We prove that the resulting theory is
unitarily equivalent to the standard "configuration space" formulation and
show that it allows for a uniform treatment of both pure and mixed quantum
states. In the second part of the paper we determine the unitary
transformation (and its infinitesimal generator) that maps the phase space
Schr\"odinger representation into another (called Moyal) representation,
where the wave function is the cross-Wigner function familiar from
deformation quantization. Some features of this representation are studied,
namely the associated pseudo-differential calculus and the main spectral and
dynamical results. Finally, the relation with deformation quantization is
discussed.
\end{abstract}

\maketitle

\section{Introduction}

A key principle of quantum mechanics states that the fundamental
configuration and momentum operators satisfy the commutation relations of
the Heisenberg algebra
\begin{equation}
\lbrack \hat{x}_{j},\hat{\xi}_{k}]=i\delta _{jk}\quad j,k=1,..,n  \label{I1}
\end{equation}%
all the other commutators being zero. The most standard implementation of
this algebra is given by the Schr\"{o}dinger representation, where $\hat{x}%
_{j}$ and $\hat{\xi}_{j}$ are viewed as self-adjoint operators
\begin{equation*}
\hat{x}_{j}=\mbox{multiplication by}\,\,x_{j}\quad ,\quad \hat{\xi}%
_{j}=-i\partial _{x_{j}}
\end{equation*}%
acting on the Hilbert space $L^{2}(\mathbb{R}^{n})$ of square integrable
functions (with support) on the classical configuration space $\mathbb{R}%
^{n} $. It is well known that the quantization rules
\begin{equation*}
x_{j}\longrightarrow \hat{x}_{j}\quad ,\quad \xi _{j}\longrightarrow \hat{\xi%
}_{j}
\end{equation*}%
do not provide the complete information on how to quantize an arbitrary
classical observable since the formal prescription $a(x,\xi
_{x})\longleftrightarrow \hat{a}=a(\hat{x},\hat{\xi}_{x})$ (where $%
x=(x_{1},...,x_{n})$, $\xi _{x}=(\xi _{1},...,\xi _{n})$) is order
ambiguous. The Weyl pseudo-differential calculus yields the standard (but
not unique) solution for this problem \cite{Weyl,Birk,Wong}. The Weyl
correspondence $a\overset{\text{Weyl}}{\longleftrightarrow }\hat{a}^{W}$ is
a one-to-one linear map that associates with each symbol $a\in S^{\prime }(%
\mathbb{R}^{2n})$ a linear operator $\hat{a}^{W}:\mathcal{S}(\mathbb{R}%
^{n})\rightarrow \mathcal{S}^{\prime }(\mathbb{R}^{n})$, uniquely defined by
\begin{equation*}
\hat{a}^{W}=\left( \frac{1}{2\pi }\right) ^{n}\int_{\mathbb{R}^{n}\times
\mathbb{R}^{n}}F_{\sigma }a(x_{0},\xi _{0})\hat{T}(x_{0},\xi
_{0})\,dx_{0}\,d\xi _{0}
\end{equation*}%
where the integral is a Bochner (operator) integral,
\begin{equation}
F_{\sigma }a(x_{0},\xi _{0})=\left( \frac{1}{2\pi }\right) ^{n}\int_{\mathbb{%
R}^{n}\times \mathbb{R}^{n}}a(x,\xi _{x})e^{i(x_{0}\cdot \xi _{x}-\xi
_{0}\cdot x)}\,dx\,d\xi _{x}  \label{ISF}
\end{equation}%
is the symplectic Fourier transform of $a$ and $\hat{T}$ is the
Heisenberg-Weyl operator, defined for all $\psi \in L^{2}(\mathbb{R}^{n})$
by
\begin{equation*}
T(x_{0},\xi _{0})\psi (x)=e^{i(\xi _{0}\cdot x-\frac{1}{2}\xi _{0}\cdot
x_{0})}\psi (x-x_{0})
\end{equation*}

We now observe that:\newline

(1) The Schr\"odinger representation is by no means the unique possible
implementation of the commutation relations (\ref{I1}). The momentum
representation in $L^2(\mathbb{R}^n)$ is a well known alternative. More
generally, any unitary operator $U: L^2(\mathbb{R}^n) \to L^2(\mathbb{R}^n)$
generates another "quantization rule" through the prescription
\begin{equation*}
x \longrightarrow U \hat x U^{-1} \quad \mbox \quad \xi_x \longrightarrow U
\hat \xi_x U^{-1}
\end{equation*}

(2) Other interesting representations can be obtained by implementing the
observables $x$ and $\xi_x$ as operators acting on Hilbert spaces, other
than $L^2(\mathbb{R}^n)$.\newline

This second possibility was seriously considered in a series of papers \cite%
{TV,GoJPA,CPDE,birkbis,GOLU1,GOLU2,Bracken} focusing on representations in
terms of operators acting on the Hilbert space $L^{2}(\mathbb{R}^{2n})$ of
functions with support on the \textit{phase space} $\mathbb{R}^{2n}$. The
Frederick and Torres-Vega representation \cite{TV,GoJPA},
\begin{equation*}
x\longrightarrow \hat{x}=x+i\frac{\partial }{\partial p}\quad ,\quad \xi
_{x}\longrightarrow \hat{\xi}_{x}=-i\frac{\partial }{\partial x}
\end{equation*}%
leading to the Schr\"{o}dinger equation in the phase space
\begin{equation*}
i\frac{\partial }{\partial t}\Psi (x,p,t)=H(x+i\frac{\partial }{\partial p}%
,-i\frac{\partial }{\partial x})\Psi (x,p,t)
\end{equation*}%
and the "Moyal" representation (given by the "Bopp shifts" \cite{Bopp})
\begin{equation}
x\longrightarrow \tilde{X}=x+\frac{1}{2}i\frac{\partial }{\partial p}\quad
,\quad \xi _{x}\longrightarrow \tilde{\Xi}_{x}=p-\frac{1}{2}i\frac{\partial
}{\partial x}  \label{IMR}
\end{equation}%
are two examples of this sort.

The latter, more symmetric, representation leads to what we shall call the
"Moyal-Weyl pseudo-differential calculus" originally presented in \cite%
{GOLU1} and further studied, in connection with the related "Landau-Weyl"
calculus, in \cite{CPDE,GOLU2}. The representation eq. (\ref{IMR}) is
intimately connected to the deformation formulation of quantum mechanics.
Indeed, for an arbitrary Weyl symbol $a\overset{\text{Weyl}}{%
\longleftrightarrow }\hat{a}^{W}$, we have in the Moyal representation \cite%
{birkbis,GOLU1}
\begin{equation*}
a(\tilde{X},\tilde{\Xi}_{x})=a(x,p)\star
\end{equation*}%
where $\star $ is the Moyal starproduct \cite{BFFLS1,Wong}. Hence, the
stargenvalue equation
\begin{equation}
a(x,p)\star \Psi _{\lambda }(x,p)=\lambda \Psi _{\lambda }(x,p)  \label{Is}
\end{equation}%
can be written in the form
\begin{equation}
a(\tilde{X},\tilde{\Xi}_{x})\Psi _{\lambda }(x,p)=\lambda \Psi _{\lambda
}(x,p)  \label{I2}
\end{equation}%
Moreover, it was also proved in \cite{birkbis,GOLU1} that the solutions of
the previous eigenvalue equation (\ref{I2}) are related to the solutions of
the usual eigenvalue equation
\begin{equation}
a(\hat{x},\hat{\xi}_{x})\psi _{\lambda }(x)=\lambda \psi _{\lambda }(x)
\label{Ie}
\end{equation}%
by the action of intertwining operators $W_{\phi }:L^{2}(\mathbb{R}%
^{n})\rightarrow L^{2}(\mathbb{R}^{2n})$ defined for each $\phi \in \mathcal{%
S}(\mathbb{R}^{n})$ by
\begin{equation*}
\Psi (x,p)=W_{\phi }\psi (x,p)=\left( \frac{1}{2\pi }\right) ^{n/2}\int_{%
\mathbb{R}^{n}}e^{-ip\cdot y}\psi (x+\frac{1}{2}y)\phi ^{\ast }(x-\frac{1}{2}%
y)\,dy
\end{equation*}%
and related to the cross Wigner distribution $W(\psi ,\phi )$ by a simple
normalization factor
\begin{equation*}
W_{\phi }\psi =(2\pi )^{n/2}W(\psi ,\phi ).
\end{equation*}%
Combining the two results (the equality of the two equations (\ref{Is}) and (%
\ref{I2}) and the relation - given by $W_{\phi }$ - of the solutions of eq.(%
\ref{I2}) with those of eq.(\ref{Ie})), de Gosson and Luef \cite{GOLU1}
found a simple proof of the spectral relation between Schr\"{o}dinger
quantum mechanics and the deformation quantization of Bayen \textit{et al.}
\cite{BFFLS1,BFFLS2}.

Some of these results were generalized in \cite{DGLP2} where an extension of
the "phase space Moyal representation" (\ref{IMR}) and the associated
pseudo-differential calculus, allowed for the precise construction of the
quantum theory associated with the extended Heisenberg algebra. The
resulting "noncommutative quantum mechanics" displays an extra
noncommutative structure in both the configurational and momentum sectors
and has played an important role in some recent approaches to quantum
cosmology \cite{babedipr,babedipr1,bracz,cahakola} and quantum gravity \cite%
{Douglas,Szabo}. Furthermore, the relation of the approach of \cite{DGLP2}
with the deformation formulation of noncommutative quantum mechanics \cite%
{dipra1,badipr} was studied in \cite{digoprlu1}.

In this paper we intend to further study the structure and the properties of
the phase space formulation of quantum mechanics, focusing, this time, on
the most central (and simplest) Schr\"odinger representation
\begin{equation*}
x \longrightarrow \hat X=x \quad , \quad \xi_x \longrightarrow \hat \Xi_x=-i%
\frac{\partial}{\partial x}
\end{equation*}
(where the operators $\hat X,\hat\Xi_x$ act on phase space functions $%
\Psi(x,p)\in L^2(\mathbb{R}^{2n})$) and on its relation with the Moyal
representation (\ref{IMR}). In addition, we will also discuss some features
of the phase space formulation of mixed quantum states.

More precisely, we will:

(1) Present the phase space formulation of quantum mechanics in the
Schr\"odinger representation.

(2) Determine and study the associated Weyl pseudo-differential calculus.

(3) Study the relation of the "Schr\"odinger phase space representation"
with the "Schr\"odinger configuration space representation". In particular,
show that there is a family of isometries $T_\chi:L^2(\mathbb{R}^n) \to L^2(%
\mathbb{R}^{2n})$ (indexed by $\chi \in L^2(\mathbb{R}^n):\, ||\chi||=1$)
that intertwine an arbitrary configuration space operator $\hat a$ with the
corresponding phase space operator $\hat A$.

(4) Show that the phase space formulation of quantum mechanics allows for an
uniform treatment of both pure and mixed quantum states.

(5) Prove that the Schr\"odinger phase space representation and the Moyal
phase space representation are unitarily related. Determine the
one-parameter group of unitary transformations (and its infinitesimal
generators) that connects the two representations and use it to prove the
main spectral and dynamical results of the Moyal representation.

(6) Discuss the relation of the Moyal representation with the deformation
quantization of Bayen \textit{et al}.

\subsection{Motivation: Double phase space formulation of classical mechanics%
}

Let us consider a dynamical system living on the phase space $\mathbb{R}^n
\oplus \mathbb{R}^n$ spanned by the canonical variables $(x,\xi_x)$ which
satisfy the usual Poisson bracket structure $\{x,\xi_x\}=I$ (where $I$ is
the $n\times n$ identity matrix) all the others being zero. Let $h(x,\xi_x)$
be the Hamiltonian of the system.

A trivial formulation of this system in the double phase space is obtained
by considering the extension ($(p,\xi_p) \in \mathbb{R}^n \oplus \mathbb{R}%
^n $)
\begin{equation*}
(x, \xi_x) \longrightarrow (x,p,\xi_x,\xi_p): \left\{
\begin{array}{l}
\{x,\xi_x\}=I \\
\{p,\xi_p\}=I%
\end{array}
\right.
\end{equation*}
(all the other commutators being zero) and the new Hamiltonian
\begin{equation*}
H_1(x,p,\xi_x,\xi_p)=h(x,\xi_x)
\end{equation*}
subjected to the initial data constraints
\begin{equation}
p=p_0 \quad , \quad \xi_p=\xi_{p0}  \label{Icc}
\end{equation}
The new set of observables
\begin{equation}
A_1(x,p,\xi_x,\xi_p)=a(x,\xi_x)  \label{IObs}
\end{equation}
yields exactly the same predictions as the original formulation in terms of
the standard phase space observables $a(x,\xi_x)$. Notice that the role of
the constraints here is only to fix the initial values of the non-physical
sector of the theory. The constraints commute with the new Hamiltonian and
are thus preserved through the time evolution.

Another double phase space formulation of classical mechanics is obtained
from the action of the symplectic transformation
\begin{equation}
S:\mathbb{R}^{2n}\oplus \mathbb{R}^{2n}\longrightarrow \mathbb{R}^{2n}\oplus
\mathbb{R}^{2n},\quad \left\{
\begin{array}{l}
x\longrightarrow x-\xi _{p}/2 \\
p\longrightarrow p-\xi _{x}/2 \\
\xi _{x}\longrightarrow \xi _{x}/2+p \\
\xi _{p}\longrightarrow \xi _{p}/2+x%
\end{array}%
\right.  \label{I3}
\end{equation}%
on the previous double phase space formulation. The classical system is then
described in terms of the new observables
\begin{equation*}
A_{2}(x,p,\xi _{x},\xi _{p})=A_{1}(x-\xi _{p}/2,p-\xi _{x}/2,\xi
_{x}/2+p,\xi _{p}/2+x)
\end{equation*}%
and the new double phase space Hamiltonian
\begin{eqnarray*}
H_{2}(x,p,\xi _{x},\xi _{p}) &=&H_{1}(x-\xi _{p}/2,p-\xi _{x}/2,\xi
_{x}/2+p,\xi _{p}/2+x) \\
&=&h(x-\xi _{p}/2,\xi _{x}/2+p)
\end{eqnarray*}%
which yield, of course, the same physical predictions as the original double
phase space formulation in terms of the Hamiltonian $H_{1}$ and the
observables $A_{1}$.

Loosely speaking, this paper is devoted to presenting the quantum
counterparts of the two former double phase space formulations of classical
mechanics, to studying their properties and their relation with the
standard, configuration space formulation of quantum mechanics. In spite of
their apparently simple structure, they yield quantum theories which display
a set of quite interesting properties, namely an uniform description of pure
and mixed states and a remarkable connection with deformation quantization.

Before finishing this section, we note for future reference that the
transformation $S$ can be written as the composition of the two following
symplectic maps
\begin{equation*}
S=S_R(\pi/4)\circ S_D(\ln \sqrt{2})
\end{equation*}
where
\begin{equation*}
S_D(s): \mathbb{R}^{4n} \longrightarrow \mathbb{R}^{4n}
\end{equation*}
\begin{equation*}
S_D(s) (x,p,\xi_x,\xi_p)= (e^s x,e^s p,e^{-s}\xi_x,e^{-s}\xi_p)\quad , \quad
s \in \mathbb{R }
\end{equation*}
and $S_R(\theta)$ is a double rotation in the $(x,\xi_p)$ and the $(p,\xi_x)$
planes
\begin{equation*}
S_R(\theta): \mathbb{R}^{4n} \longrightarrow \mathbb{R}^{4n}
\end{equation*}
\begin{equation*}
S_R(\theta): \left\{
\begin{array}{l}
(x,\xi_p) \longrightarrow (x \cos \theta - \xi_p \sin \theta , x \sin \theta
+\xi_p \cos \theta) \\
(p,\xi_x) \longrightarrow (p \cos \theta - \xi_x \sin \theta , p \sin \theta
+\xi_x \cos \theta)%
\end{array}
\right.
\end{equation*}
The infinitesimal generators of $S_D$ and $S_R$ are the Hamiltonians
\begin{equation}
H_D=x \cdot \xi_x +p \cdot \xi_p  \label{IHD}
\end{equation}
and
\begin{equation}
H_R= -\xi_x \cdot \xi_p -x \cdot p  \label{IHR}
\end{equation}
respectively.

\subsection{Notation}

A generic point of the original phase space $\mathbb{R}^{2n}=\mathbb{R}%
^{n}\oplus \mathbb{R}^{n}$ is denoted by $z_{x}=(x,\xi _{x})$ and that of
the double phase space $\mathbb{R}^{4n}=\mathbb{R}^{2n}\oplus \mathbb{R}%
^{2n} $ by $Z=(z_{x},z_{p})$ where $z_{p}=(p,\xi _{p})$. We will also use $%
z=(x,p)$ and $z_{0}=(x_{0},\xi _{x0})$. The symplectic form on $\mathbb{R}%
^{2n}$ is $\sigma (z_{x},z_{x}^{\prime })=\xi _{x}\cdot x^{\prime }-\xi
_{x}^{\prime }\cdot x$ and on $\mathbb{R}^{2n}\oplus \mathbb{R}^{2n}$ is $%
\sigma _{x}\oplus \sigma _{p}(Z,Z^{\prime })=\sigma (z_{x},z_{x}^{\prime
})+\sigma (z_{p},z_{p}^{\prime })$. The corresponding symplectic groups are
denoted by $Sp(2n,\sigma )$ and $Sp(4n,\sigma _{x}\oplus \sigma _{p})$.

We write $\mathcal{S}(\mathbb{R}^n)$ for the Schwartz space of rapidly
decreasing test functions on $\mathbb{R}^n$ and $\mathcal{S}^{\prime }(%
\mathbb{R}^n)$ for the dual space of tempered distributions. The notation $%
\langle a,t\rangle $ stands for the action of the distribution $a$ on the
test function $t$. Symbols (or classical observables) on $\mathbb{R}^{2n}$
are denoted by small Latin letters $a,b,..$; if they have support on $%
\mathbb{R}^{4n}$ they are denoted by capital Latin letters $A,B,...$. The
wave functions in $L^2(\mathbb{R}^n)$ are denoted by small Greek letters $%
\psi,\phi,...$ and those in $L^2(\mathbb{R}^{2n})$ by capital Greek letters $%
\Psi,\Phi,...$. The standard inner product on $L^2(\mathbb{R}^{n})$ is
written $(\psi|\phi)$; on $L^2(\mathbb{R}^{2n})$ is $((\Psi|\Phi))$. The
corresponding norms are $||\psi||$ and $|||\Psi|||$.

Operators acting on functions (or distributions) on $\mathbb{R}^{n}$ are
usually denoted by small Latin letters with a hat $\hat{a},\hat{b},...$ and
those acting on phase space functions or distributions usually by $\hat{A},%
\hat{B},...$ if they are in the Schr\"{o}dinger representation and by $%
\tilde{A},\tilde{B},...$ if they are in the Moyal representation. Weyl
pseudo-differential operators display a $W$-superscript. The superscript $%
\ast $ denotes both the complex conjugation (for functions) and the adjoint
(for operators).

The unitary Fourier transform and its inverse on $L^{2}(\mathbb{R}^{n})$ are
defined by
\begin{equation*}
\hat{\phi}(p)=\mathcal{F}_{x}[\phi (x)](p)=\frac{1}{(2\pi )^{n/2}}\int_{%
\mathbb{R}^{n}}e^{-ix\cdot p}\phi (x)\,dx
\end{equation*}%
and
\begin{equation*}
\mathcal{F}_{x}^{-1}[\phi (p)](x)=\frac{1}{(2\pi )^{n/2}}\int_{\mathbb{R}%
^{n}}e^{ix\cdot p}\phi (p)\,dp.
\end{equation*}%
The same notation will also be used for the generalized (i.e.
distributional) Fourier transform.

\section{Phase space quantum mechanics in the Schr\"odinger representation:
General results}

The key object relating the configuration and the phase space Schr\"odinger
representations of the Heisenberg algebra is the map
\begin{equation}
T_{\chi} : L^2(\mathbb{R}^n) \longrightarrow L^2(\mathbb{R}^{2n}); \quad
\psi \to T_{\chi}[\psi]:= \psi \otimes \chi^*  \label{II1}
\end{equation}
which is defined for some fixed $\chi \in L^2(\mathbb{R}^n)$ such that $%
||\chi||=1$. Since, for all $\phi,\psi\in L^2(\mathbb{R}^n)$
\begin{equation*}
((T_{\chi}[\psi]\,|\, T_{\chi}[\phi]))=(\psi\,|\,\phi)
\end{equation*}
$T_{\chi}$ is an isometry from $L^2(\mathbb{R}^n)$ into a subspace ${%
\mathcal{H}}_\chi$ of $L^2(\mathbb{R}^{2n})$. Hence the map $T_\chi$ is also
linear, injective and continuous. We then have the following

\begin{definition}[States]

The states of the new (phase space) formulation of quantum mechanics are the
wave functions
\begin{equation}
\Psi(x,p)=T_{\chi}[\psi](x,p)=\psi(x) \chi^*(p)  \label{II2}
\end{equation}
which form the Hilbert space ${\mathcal{H}}_\chi =$Ran $T_\chi \subset L^2(%
\mathbb{R}^{2n})$ (Corollary 2.5, bellow).
\end{definition}

Notice that:

\begin{remark}
The choice of a particular $\chi$ plays here the same role as the imposition
of the classical constraints (\ref{Icc}) in the classical double phase space
formulation. In fact, just like in the classical case, this imposition
tantamount to the complete specification of the initial data for the
unphysical sector of the theory. Notice also that the quantum states do not
satisfy the strong (Dirac) version of the quantum constraints \cite{Henneaux}
\begin{equation*}
(\hat p-p_0)\Psi=0 \quad , \quad (\hat\xi_p-\xi_{p0}) \Psi=0
\end{equation*}
which are, in fact, incompatible (because $\hat p,\hat\xi_p$ do not
commute). They may, however satisfy a weaker version
\begin{equation*}
((\Psi|(\hat p-p_0)\Psi))=0 \quad , \quad ((\Psi|(\hat \xi_p-\xi_{p0}))
\Psi))=0
\end{equation*}
which may be seen as a necessary, but not sufficient, condition for the
states to be of the form (\ref{II2}).
\end{remark}

A key property of $T_{\chi }$ is given by:

\begin{theorem}
The adjoint $T_{\chi }^{\ast }$ is the map
\begin{equation*}
T_{\chi }^{\ast }:L^{2}(\mathbb{R}^{2n})\longrightarrow L^{2}(\mathbb{R}^{n})
\end{equation*}%
\begin{equation*}
\Psi (x,p)\longrightarrow T_{\chi }^{\ast }[\Psi ](x)=\int_{\mathbb{R}%
^{n}}\Psi (x,p)\chi (p)\,dp.
\end{equation*}
\end{theorem}

\begin{proof}
Consider the most general states
\begin{equation*}
(\psi (x),\Psi (x,p))\in L^{2}(\mathbb{R}^{n})\times L^{2}(\mathbb{R}^{2n})
\end{equation*}%
for which
\begin{equation}
(\psi \,|\,\phi )=((\Psi \,|\,T_{\chi }[\phi ])),\quad \forall \phi \in
L^{2}(\mathbb{R}^{n}).  \label{II3}
\end{equation}%
Then:

(1) $\Psi$ belongs to the domain of the adjoint $D(T^*_{\chi})$ and

(2) $T^*_{\chi} [\Psi]=\psi$.\newline

Eq.(\ref{II3}) is equivalent to
\begin{eqnarray}
&&(\psi\,|\,\phi)=((\Psi\,|\,\phi \otimes \chi^*)), \quad \forall \phi \in
L^2(\mathbb{R}^n)  \notag \\
&\Longleftrightarrow & \int_{\mathbb{R}^n} \phi(x)\psi^*(x)\, dx= \int_{%
\mathbb{R}^n} \, \phi(x) \left[ \int_{\mathbb{R}^n} \Psi^*(x,p) \chi^*(p) \,
dp \right]\, dx, \quad \forall \phi \in L^2(\mathbb{R}^n)  \notag \\
&\Longleftrightarrow & \psi(x) = \int_{\mathbb{R}^n} \Psi(x,p) \chi(p)\, dp
\notag
\end{eqnarray}
It also follows that $\psi \in L^2(\mathbb{R}^n)$ if $\Psi \in L^2(\mathbb{R}%
^{2n})$ and so $D(T^*_{\chi})=L^2(\mathbb{R}^{2n})$, which concludes the
proof.
\end{proof}

\begin{remark}
Let $\mathcal{H}_\chi=$Ran $T_\chi$. The restriction
\begin{equation*}
T^*_\chi:\mathcal{H}_\chi \longrightarrow L^2(\mathbb{R}^n)
\end{equation*}
is a one to one map and is the inverse of $T_\chi$. The proof is trivial.
\end{remark}

Moreover:

\begin{corollary}
The operator
\begin{equation*}
P_{\chi}:L^2(\mathbb{R}^{2n}) \longrightarrow L^2(\mathbb{R}^{2n}),\quad
P_{\chi}:=T_{\chi} \circ T^*_{\chi}
\end{equation*}
is a projector and
\begin{equation*}
\mathcal{H}_{\chi}=\mbox{Ran}\, T_{\chi}=\mbox{Ran}\, P_{\chi}
\end{equation*}
is a Hilbert space, Hilbert subspace of $L^2(\mathbb{R}^{2n})$, with the
inner product
\begin{equation}
((T_{\chi}[\phi]\,|\,T_{\chi}[\psi]))=(\phi\,|\,\psi)  \label{II4}
\end{equation}
\end{corollary}

\begin{proof}
$P_{\chi}$ is an orthogonal projector because (i) $P_{\chi}=P_{\chi}^*$ and
(ii) $P_{\chi}P_{\chi}=P_{\chi}$. The first identity follows from $%
T_{\chi}^{**}=T_{\chi}$:
\begin{equation*}
((\Phi\,|\,P_{\chi}[\Psi]))=((\Phi\,|\,T_{\chi}T^*_{\chi}[\Psi]))=(T^*_{%
\chi}[\Phi]\,|\,T^*_{\chi}[\Psi])
\end{equation*}
\begin{equation*}
=((T_{\chi}T^*_{\chi}[\Phi]\,|\,\Psi))=((P_{\chi}[\Phi]\,|\,\Psi))\quad ,
\quad \forall \Phi,\Psi \in L^2(\mathbb{R}^{2n})
\end{equation*}
To prove (ii) just notice that $T^*_{\chi}T_{\chi}$ is the identity on $L^2(%
\mathbb{R}^n)$.

Since the range of $T^*_{\chi}$ is $L^2(\mathbb{R}^n)$ which is the domain
of $T_{\chi}$ (and $P_{\chi}=T_{\chi}T^*_{\chi}$) we have Ran $T_{\chi}=$Ran
$P_{\chi}$. On the other hand, the range of a projector in a Hilbert space
is a closed linear subspace of that Hilbert space, thus also a Hilbert
space. Hence, ${\mathcal{H}}_{\chi}=$Ran $P_{\chi}$ is a Hilbert space,
subspace of $L^2(\mathbb{R}^{2n})$ and with the same inner product. The
identity (\ref{II4}) follows trivially .
\end{proof}

\begin{remark}
Notice that there are many possible choices for the space of states ${%
\mathcal{H}}_{\chi}$. This is related, of course, to the possibility of
choosing $\chi$ arbitrarily in $L^2(\mathbb{R}^n)$. However, these
(different) phase space formulations are related by the one to one,
orthogonal transformations
\begin{equation*}
T_{\chi_1} T^*_{\chi_2}:{\mathcal{H}}_{\chi_2} \longrightarrow {\mathcal{H}}%
_{\chi_1}
\end{equation*}
and are thus completely equivalent.
\end{remark}

We now introduce operators on ${\mathcal{H}}_{\chi}$:

\begin{definition}[Phase space representation]

For each "configuration space" operator
\begin{equation*}
\hat a:D(\hat a)\subset L^2(\mathbb{R}^n) \longrightarrow L^2(\mathbb{R}^n)
\end{equation*}
we define an operator on ${\mathcal{H}}_{\chi}$
\begin{equation}
\hat A:D(\hat A) \subset {\mathcal{H}}_{\chi} \longrightarrow {\mathcal{H}}%
_{\chi}  \label{II5}
\end{equation}
given by
\begin{equation*}
D(\hat A)=T_{\chi}[D(\hat a)]
\end{equation*}
\begin{equation*}
\hat A \Psi =T_{\chi} \hat a T_{\chi}^* \Psi \quad , \quad \forall \Psi \in
D(\hat A)
\end{equation*}
and call $\hat A$ the "phase space operator associated with $\hat a$" or,
more precisely, the "${\mathcal{H}}_{\chi}$-representation of $\hat a$".
\end{definition}

\begin{remark}
Notice that for
\begin{equation*}
\hat{a}=\hat{x}_{i}=\mbox{multiplication by}\,\,x_{i}\,\,\mbox{in}\,\,L^{2}(%
\mathbb{R}^{n})
\end{equation*}%
the associated operator is
\begin{equation*}
\hat{A}=\hat{X}_{i}=T_{\chi }\hat{x}_{i}T_{\chi }^{\ast }=%
\mbox{multiplication
by}\,\,x_{i}\,\,\mbox {in}\,\,{\mathcal{H}}_{\chi }\subset L^{2}(\mathbb{R}%
^{2n})
\end{equation*}%
and for
\begin{equation*}
\hat{a}=\hat{\xi}_{i}=-i\partial _{x_{i}}\quad (\mbox{in}\,\,L^{2}(\mathbb{R}%
^{n}))
\end{equation*}%
we have
\begin{equation*}
\hat{A}=\hat{\Xi}_{i}=T_{\chi }\hat{\xi}_{i}T_{\chi }^{\ast }=-i\partial
_{x_{i}}\quad (\mbox{in}\,\,{\mathcal{H}}_{\chi }\subset L^{2}(\mathbb{R}%
^{2n})).
\end{equation*}%
Hence, each map $T_{\chi }$ generates a phase space Schr\"{o}dinger
representation of the Heisenberg algebra from the original configuration
space Schr\"{o}dinger representation.
\end{remark}

We proceed with the study of some properties of the operators $\hat A$.

\begin{theorem}
Let $\hat A$ be the phase space operator associated with $\hat a$ in the
sense of the Definition 2.7. Then

(1) $\hat A$ is symmetric iff $\hat a$ is symmetric.

(2) $\hat A$ is self-adjoint (as an operator in $\mathcal{H}_\chi$) iff $%
\hat a$ is self-adjoint.
\end{theorem}

\begin{proof}
1) From
\begin{equation*}
\hat A \Psi=T_{\chi} \hat a T_{\chi}^*\, T_{\chi}[\psi]=T_{\chi}[\hat a
\psi], \quad \forall \, \Psi=T_{\chi}[\psi] \in D(\hat A)
\end{equation*}
it follows that
\begin{equation*}
((\hat A\Psi\,|\,\Phi))=((\Psi\,|\,\hat A \Phi)), \quad \forall \Psi,\Phi
\in D(\hat A)
\end{equation*}
\begin{equation*}
\Longleftrightarrow ((T_{\chi}[\hat a
\psi]\,|\,T_{\chi}[\phi]))=((T_{\chi}[\psi]\,|\,T_{\chi}[\hat a \phi])),
\quad \forall \psi,\phi \in D(\hat a)
\end{equation*}
\begin{equation*}
\Longleftrightarrow (\hat a \psi\,|\,\phi)=(\psi\,|\,\hat a \phi) ,\quad
\forall \psi,\phi \in D(\hat a)
\end{equation*}
and so $\hat A$ is symmetric in $D(\hat A)$ iff $\hat a$ is symmetric in $%
D(\hat a)$.\newline

2) Let $\hat{a}$ be defined on a dense domain $D(\hat{a})\subset L^{2}(%
\mathbb{R}^{n})$ and consider the most general solution $(\psi ,\xi )\in
L^{2}(\mathbb{R}^{n})\times L^{2}(\mathbb{R}^{n})$ of
\begin{equation}
(\xi \,|\,\phi )=(\psi \,|\,\hat{a}\phi ),\quad \forall \phi \in D(\hat{a})
\end{equation}%
Then $\psi \in D(\hat{a}^{\ast })$ and $\hat{a}^{\ast }\psi =\xi $. For $%
\Psi =T_{\chi }[\psi ]\Longleftrightarrow \psi =T_{\chi }^{\ast }[\Psi ]$
and $\Xi =T_{\chi }[\xi ]\Longleftrightarrow \xi =T_{\chi }^{\ast }[\Xi ]$,
the previous equation is equivalent to
\begin{equation*}
((\Xi \,|\,\Phi ))=((\Psi \,|\,\hat{A}\Phi )),\quad \forall \Phi \in D(\hat{A%
})
\end{equation*}%
and so (since $D(\hat{A})=T_{\chi }[D(\hat{a})]\subset \mathcal{H}_{\chi }$
is also dense) $\Psi \in D(\hat{A}^{\ast })$ and $\hat{A}^{\ast }\Psi =\Xi $%
. Hence
\begin{equation*}
D(\hat{A}^{\ast })=T_{\chi }[D(\hat{a}^{\ast })]\quad ,\quad \hat{A}^{\ast
}=T_{\chi }\hat{a}^{\ast }T_{\chi }^{\ast }.
\end{equation*}%
It follows that $\hat{a}=\hat{a}^{\ast }\Longleftrightarrow \hat{A}=\hat{A}%
^{\ast }$ as claimed.
\end{proof}

\begin{theorem}[Spectral results]

The state $\Psi_\lambda \in {\mathcal{H}}_{\chi}$ is a solution of the
eigenvalue equation
\begin{equation}
\hat A \Psi_\lambda=\lambda\Psi_\lambda  \label{IeeA}
\end{equation}
iff $\Psi_\lambda=T_{\chi}[\psi_\lambda]$ for some $\psi_\lambda\in L^2(%
\mathbb{R}^n)$ satisfying
\begin{equation}
\hat a \psi_\lambda=\lambda\psi_\lambda  \label{Ieea}
\end{equation}
Hence, the two operators $\hat A$ and $\hat a$ display the same spectrum and
their eigenfunctions are related by the $T_{\chi}$-transformation.
\end{theorem}

\begin{proof}
The proof is trivial. In one direction it follows from applying $T_{\chi }$
to both sides of the eigenvalue equation (\ref{Ieea}) and, in the other
direction, from applying $T_{\chi }^{\ast }$ to both sides of (\ref{IeeA}).
We also need to notice that
\begin{equation}
T_{\chi }[\hat{a}\psi ]=\hat{A}T_{\chi }[\psi ]=\hat{A}\Psi  \label{II-I1}
\end{equation}%
and
\begin{equation}
T_{\chi }^{\ast }[\hat{A}\Psi ]=\hat{a}T_{\chi }^{\ast }[\Psi ]=\hat{a}\psi .
\label{II-I2}
\end{equation}
\end{proof}

\begin{corollary}
We conclude that the physical predictions of the phase space and the
configuration space formulations of quantum mechanics are the same. For the
average values we have from Definitions 2.1 and 2.7
\begin{equation*}
((\Psi \,|\,\hat{A}\Phi ))=(\psi \,|\,\hat{a}\phi )
\end{equation*}%
and for the transition amplitudes (and probability amplitudes) from Theorem
2.10
\begin{equation*}
((\Psi \,|\,\Psi _{\lambda }))=(\psi \,|\,\psi _{\lambda }).
\end{equation*}
\end{corollary}

The same conclusion is valid for the dynamics,

\begin{theorem}[Dynamics]

Let $\hat{a}:D(\hat{a})\subset L^{2}(\mathbb{R}^{n})\rightarrow L^{2}(%
\mathbb{R}^{n})$ be self-adjoint. Then $\psi (\cdot ,t)\in L^{2}(\mathbb{R}%
^{n})$ is the solution of the initial value problem
\begin{equation*}
i\frac{\partial \psi }{\partial t}=\hat{a}\psi \quad ,\quad \psi (\cdot
,0)=\psi _{0}\in D(\hat{a})
\end{equation*}%
if and only if $\Psi =T_{\chi }[\psi ]$ is the solution of
\begin{equation*}
i\frac{\partial \Psi }{\partial t}=\hat{A}\Psi \quad ,\quad \Psi (\cdot
,0)=T_{\chi }[\psi _{0}]\in D(\hat{A}).
\end{equation*}
\end{theorem}

\begin{proof}
This theorem is also trivial. Again the equivalence of the two dynamics can
be proved by applying the maps $T_{\chi}$ and $T^*_{\chi}$ to the dynamical
equations and by taking into account the relations (\ref{II-I1},\ref{II-I2})
and also that the time derivative commutes with $T_{\chi}$ (because $\chi$
is time independent)
\begin{equation*}
\frac{\partial }{\partial t}T_{\chi}[\psi]=T_{\chi}[\frac{\partial }{%
\partial t}\psi]
\end{equation*}
and with $T^*_{\chi}$
\begin{equation*}
\frac{\partial }{\partial t}T^*_{\chi}[\Psi]=T^*_{\chi}[\frac{\partial }{%
\partial t}\Psi]
\end{equation*}
\end{proof}

\section{Mixed states in phase space quantum mechanics}

The formalism of the last section was defined in the spaces ${\mathcal{H}}%
_{\chi}$ and provides a description of pure states, only. However, one
realizes that most of the results can be generalized for the case where the
space of states is a Hilbert space larger than a particular $\mathcal{H}%
_\chi $ or is even the entire $L^2(\mathbb{R}^{2n})$.

Such generalizations lead to a suitable description of mixed states as we
now show. More precisely, we will discuss (some aspects of) the extension of
the formalism of the last section to the case where the Hilbert space of
states is of the form
\begin{equation*}
{\mathcal{H}}=\oplus _{k}{\mathcal{H}}_{\chi _{k}}
\end{equation*}%
for some orthogonal, finite set of functions $\chi _{k}\in L^{2}(\mathbb{R}%
^{n})$ such that
\begin{equation*}
\sum_{k}||\chi _{k}||^{2}=1.
\end{equation*}%
Then
\begin{equation*}
\Psi \in {\mathcal{H}}\Longrightarrow \Psi =\sum_{k}\psi _{k}\otimes \chi
_{k}^{\ast },\quad \psi _{k}\in L^{2}(\mathbb{R}^{n})
\end{equation*}%
The normalization of $\Psi $ follows directly from the normalization of each
$\psi _{k}$
\begin{equation*}
|||\Psi |||^{2}=((\Psi \,|\,\Psi ))=\sum_{k}(\psi _{k}\,|\,\psi _{k})(\chi
_{k}\,|\,\chi _{k})=\sum_{k}||\psi _{k}||^{2}||\chi _{k}||^{2}
\end{equation*}%
and this formula suggests that the square of the norm of $\chi _{k}$ can be
interpreted as the classical probability associated with the $k$-component
of the mixed state $\Psi $.

The ${\mathcal{H}}$-representation of a configuration space operator $\hat{a}
$ generalizes Definition 2.7. We now have
\begin{equation*}
\hat{A}:D(\hat{A})\subset {\mathcal{H}}\longrightarrow {\mathcal{H}}
\end{equation*}%
where the domain is
\begin{equation*}
D(\hat{A})=\oplus _{k}T_{\chi _{k}}[D(\hat{a})]
\end{equation*}%
and%
\begin{equation*}
\hat{A}\Psi =\sum_{k}T_{\chi _{k}}\hat{a}T_{\chi _{k}}^{\ast }\,\Psi \quad
,\quad \forall \Psi \in D(\hat{A}).
\end{equation*}%
Let us now recover the quantum predictions for mixed states from applying
the standard rules of "pure state quantum mechanics" to the phase space ${%
\mathcal{H}}$-formulation. Let $\phi _{\alpha }$ be a normalized
eigenfunction of $\hat{a}$ associated with the non-degenerate eigenvalue $%
\alpha $ (to make it simpler we shall assume the one-dimensional case; $n=1$%
). Then $\alpha $ is also an eigenvalue of $\hat{A}$ (in this case,
degenerate). A normalized basis of the $\alpha $-eigenspace of $\hat{A}$ is
given by the eigenfunctions
\begin{equation*}
\Phi _{\alpha ,k}=\phi _{\alpha }\otimes \frac{\chi _{k}^{\ast }}{||\chi
_{k}||},
\end{equation*}%
hence the probability associated with the eigenvalue $\alpha $ is (from
standard rules)
\begin{equation*}
{\mathcal{P}}(\hat{A}=\alpha )=\sum_{k}\left\vert ((\Psi \,|\,\Phi _{\alpha
,k}))\right\vert ^{2}=\sum_{k}\left\vert (\psi _{k}\,|\,\phi _{\alpha
})\right\vert ^{2}||\chi _{k}||^{2}
\end{equation*}%
which is then a convex combination of the probabilities ${\mathcal{P}}(\hat{a%
}=\alpha )$ calculated for each of the components $\psi _{k}$ of the mixed
state $\Psi $. This result thus re-enforces the interpretation of the square
of the norm of $\chi _{k}$ as the classical probability associated with the $%
k$-component of the mixed state.

The projection of $\Psi $ into the $\alpha $-eigenspace of $\hat{A}$ yields
the state
\begin{equation*}
\Upsilon =\sum_{k}((\Psi \,|\,\Phi _{\alpha ,k}))\Phi _{\alpha
,k}=\sum_{k}(\psi _{k}\,|\,\phi _{\alpha })\phi _{\alpha }\otimes \chi
_{k}^{\ast }
\end{equation*}%
and so the collapse of the wave function (produced by a measurement of $\hat{%
A}$ with output $\alpha $) yields
\begin{equation*}
\Upsilon _{c}=\frac{\Upsilon }{|||\Upsilon |||}.
\end{equation*}%
Just like for standard, pure state, quantum mechanics the probability ${%
\mathcal{P}}(\hat{A}=\alpha )$ is identical to the transition probability $%
|((\Psi \,|\,\Upsilon _{c}))|^{2}$. Indeed
\begin{equation*}
|((\Psi \,|\,\Upsilon _{c}))|^{2}=\frac{1}{|||\Upsilon |||^{2}}%
|((\,\sum_{k}\psi _{k}\otimes \chi _{k}^{\ast }\,|\,\sum_{j}(\psi
_{j}\,|\,\phi _{\alpha })\phi _{\alpha }\otimes \chi _{j}^{\ast }\,))|^{2}
\end{equation*}%
and since
\begin{equation*}
|||\Upsilon |||^{2}=\sum_{k}|(\psi _{k}\,|\,\phi _{a})|^{2}||\chi _{k}||^{2}
\end{equation*}%
we get
\begin{equation*}
|((\Psi \,|\,\Upsilon _{c}))|^{2}=\frac{1}{|||\Upsilon |||^{2}}\left\vert
\sum_{k}|(\psi _{k}\,|\,\phi _{a})|^{2}||\chi _{k}||^{2}\right\vert
^{2}=\sum_{k}|(\psi _{k}\,|\,\phi _{a})|^{2}||\chi _{k}||^{2}.
\end{equation*}%
Hence, a representation of mixed states in terms of standard wave functions
with support on the phase space is possible. We intend to study this topic
in more detail in a forthcoming paper.

\section{Phase space Weyl calculus}

We start this section with a brief review of the standard definitions and
properties of Weyl operators acting on functions with support on the
configuration space (for details and proofs the interested reader may refer
to \cite{Birk,birkbis,Hor2,Shubin,Stein,Wong}). We then present the
extension of these operators to phase space functions and study the main
properties of the resulting phase space Weyl calculus.

\subsection{Standard Weyl calculus}

Let $\mathscr L(\mathcal{S}(\mathbb{R}^{n}),\mathcal{S}^{\prime }(\mathbb{R}%
^{n}))$ be the space of linear and continuous operators of the form $%
\mathcal{S}(\mathbb{R}^{n})\longrightarrow \mathcal{S}^{\prime }(\mathbb{R}%
^{n})$. In view of Schwartz kernel theorem all operators $\hat{a}\in
\mathscr
L(\mathcal{S}(\mathbb{R}^{n}),S^{\prime }(\mathbb{R}^{n}))$ admit a kernel
representation
\begin{equation}
\hat{a}\psi (x)=\langle K_{a}(x,\cdot ),\psi (\cdot )\rangle  \label{IV1}
\end{equation}%
where $K_{a}\in \mathcal{S}^{\prime }(\mathbb{R}^{n}\times \mathbb{R}^{n})$.
The Weyl symbol of $\hat{a}$ is then
\begin{equation}
a(x,\xi _{x})=\int_{\mathbb{R}^{n}}e^{-i\xi _{x}\cdot y}K_{a}(x+\frac{1}{2}%
y,x-\frac{1}{2}y)\,dy  \label{IV2}
\end{equation}%
and, conversely,%
\begin{equation}
K_{a}(x,y)=\left( \frac{1}{2\pi }\right) ^{n}\int_{\mathbb{R}^{n}}e^{i\xi
\cdot (x-y)}a(\frac{x+y}{2},\xi )\,d\xi  \label{IV3}
\end{equation}%
where the integrals are interpreted as generalized Fourier (and inverse
Fourier) transforms (i.e. in the sense of distributions). The inverse
formula can be re-written in the form
\begin{equation}
K_{a}(x,y)=\left( \frac{1}{2\pi }\right) ^{n}\int_{\mathbb{R}^{n}}e^{\frac{i%
}{2}\xi \cdot (x+y)}\mathcal{F}_{\sigma }a(x-y,\xi )\,d\xi  \label{IV4}
\end{equation}%
where $\mathcal{F}_{\sigma }$ is the symplectic Fourier transform (\ref{ISF}%
).

An important property of Weyl operators is that, if $\hat b \in \mathscr L(%
\mathcal{S}(\mathbb{R}^n),\mathcal{S}(\mathbb{R}^n))$ then $\hat c=\hat a
\hat b \in \mathscr L(\mathcal{S}(\mathbb{R}^n),\mathcal{S}^{\prime }(%
\mathbb{R}^n))$ and its Weyl symbol is given by
\begin{equation}
c(z_x)=\left(\frac{1}{4\pi}\right)^{2n}\int \int_{\mathbb{R}^{2n} \times
\mathbb{R}^{2n}} e^{\frac{i}{2} \sigma(u,v)} a(z_x+\frac{1}{2}u)b(z_x-\frac{1%
}{2}v) \, du\, dv=a(z_x) \star b(z_x)  \label{IVMSP}
\end{equation}
where $\star$ is the "twisted product" or Moyal product familiar from
deformation quantization \cite{BFFLS1,Maillard,Wong}.

The Weyl correspondence $a\overset{\text{Weyl}}{\longleftrightarrow }\hat{a}%
^{W}$ (given by eqs.(\ref{IV1},\ref{IV2})) can be written more
straightforwardly as
\begin{equation}
\hat{a}^{W}=\left( \frac{1}{2\pi }\right) ^{n}\int_{\mathbb{R}^{2n}}\mathcal{%
F}_{\sigma }a(z_{0})\hat{T}(z_{0})\,dz_{0}  \label{IVaw}
\end{equation}%
where the integral is an operator valued (Bochner) integral and $\hat{T}%
(z_{0}):L^{2}(\mathbb{R}^{n})\longrightarrow L^{2}(\mathbb{R}^{n})$ is the
Heisenberg-Weyl operator ($z_{0}=(x_{0},\xi _{x0})$)
\begin{equation}
\hat{T}(z_{0})=e^{-i\sigma (\hat{z}_{x},z_{0})}=e^{-i(x_{0}\cdot \hat{\xi}%
_{x}-\xi _{x0}\cdot \hat{x})};  \label{IVHW1}
\end{equation}%
explicitly:
\begin{equation}
\hat{T}(x_{0},\xi _{x0})\psi (x)=e^{i(\xi _{x0}\cdot x-\frac{1}{2}\xi
_{x0}\cdot x_{0})}\psi (x-x_{0}).  \label{IVHW2}
\end{equation}

\begin{remark}
The Heisenberg-Weyl operators $\hat{T}(z_{0})$ are unitary operators
associated with the self-adjoint Hamiltonians $\sigma (\hat{z}_{x},z_{0})$.
In fact, one can easily check that
\begin{equation*}
\psi (x,t)=\hat{U}(z_{0},t)\psi _{0}(x)=e^{it(\xi _{x0}\cdot x-\frac{1}{2}%
\xi _{x0}\cdot x_{0})}\psi _{0}(x-tx_{0}),\quad \psi _{0}\in \mathcal{S}(%
\mathbb{R}^{n})
\end{equation*}%
is the unique solution of the initial value problem
\begin{equation*}
i\frac{\partial }{\partial t}\psi =\sigma (\hat{z}_{x},z_{0})\psi \quad
,\quad \psi (\cdot ,0)=\psi _{0}\in \mathcal{S}(\mathbb{R}^{n})
\end{equation*}%
Hence $\hat{U}(z_{0},t)=e^{-it\sigma (\hat{z}_{x},z_{0})}$ in $\mathcal{S}(%
\mathbb{R}^{n})$ and since both operators are continuous we also have
\begin{equation*}
\hat{U}(z_{0},t)=e^{-it\sigma (\hat{z}_{x},z_{0})}\quad \mbox{in}\,L^{2}(%
\mathbb{R}^{n}).
\end{equation*}%
Finally, it is trivial to check that $\hat{T}(z_{0})=\hat{U}(z_{0},1)$.
\end{remark}

The Weyl correspondence satisfies the metaplectic covariance property, which
will be useful in section 5,

\begin{theorem}
Let $\limfunc{Mp}(2n,\sigma )$ be the metaplectic group, i.e. the unitary
representation of the double cover of $Sp(2n,\sigma )$. Let $\hat{a}^{W}%
\overset{\text{Weyl}}{\longleftrightarrow }a$, let $S\in $ $Sp(2n,\sigma )$
and let $\hat{S}\in $ $\limfunc{Mp}(2n,\sigma )$ be (one of the two)
metaplectic operators that projects onto $S$. Then
\begin{equation}
\hat{S}^{-1}\hat{a}^{W}\hat{S}\overset{\text{Weyl}}{\longleftrightarrow }%
a\circ S  \label{IVm}
\end{equation}
\end{theorem}

For proof see \cite{Stein,Wong}.

\subsection{Weyl calculus in phase space}

In this section we consider the extensions of $\hat a^W$ to $\mathcal{H}%
_\chi $ and $\mathcal{S}(\mathbb{R}^{2n})$. We will focus on the case where $%
\chi \in \mathcal{S}(\mathbb{R}^n)$ so that $T_\chi[\psi] \in \mathcal{S}(%
\mathbb{R}^{2n})$ for all $\psi \in \mathcal{S}(\mathbb{R}^{n})$.

Let us then introduce the notation
\begin{equation*}
\mathcal{S}_\chi=\mathcal{S}(\mathbb{R}^{2n}) \cap \mathcal{H}_{\chi}=T_\chi[%
\mathcal{S}(\mathbb{R}^n)]
\end{equation*}
and consider the extension of $T_\chi$ to $\mathcal{S}^{\prime }(\mathbb{R}%
^n)$
\begin{equation*}
T_{\chi}:\mathcal{S}^{\prime }(\mathbb{R}^n) \longrightarrow \mathcal{S}%
^{\prime }(\mathbb{R}^{2n}); \psi \longrightarrow
T_{\chi}[\psi]=\psi\otimes\chi^*
\end{equation*}
Then $\Psi=T_\chi[\psi]$ is the distribution
\begin{equation*}
\langle \Psi,t\rangle =\langle \psi,T^*_{\chi^*}[t]\rangle ,\quad \forall t
\in \mathcal{S}(\mathbb{R}^{2n})
\end{equation*}
which is well defined because $\chi \in \mathcal{S}(\mathbb{R}^n)$.

By a natural generalization of Definition 2.7 (to operators of the form $%
\hat a:\mathcal{S}(\mathbb{R}^n) \to \mathcal{S}^{\prime }(\mathbb{R}^{n})$)
the extension of $\hat a^W$ to $\mathcal{S}_{\chi}\subset \mathcal{S}(%
\mathbb{R}^{2n})$ is given by
\begin{equation*}
\hat A^W_\chi: \mathcal{S}_{\chi} \longrightarrow \mathcal{S}^{\prime }(%
\mathbb{R}^{2n})
\end{equation*}
\begin{equation*}
\hat A^W_\chi=T_{\chi} \hat a^W T_{\chi}^*=\left(\frac{1}{2\pi}\right)^n
\int_{\mathbb{R}^{2n}} \mathcal{F}_{\sigma}a(z_0) T_{\chi} \hat T(z_0)
T_{\chi}^* \, dz_0
\end{equation*}
and its action is explicitly
\begin{equation*}
\hat A^W_\chi \Psi(x,p)=\left(\frac{1}{2\pi}\right)^n \langle \mathcal{F}%
_{\sigma}a(\cdot),T_{\chi} \hat T(\cdot) T_{\chi}^*\Psi(x,p)\rangle
\end{equation*}
which is well defined since
\begin{equation*}
T_{\chi} \hat T(x_0,\xi_{x0}) T_{\chi}^*\Psi(x,p)=e^{i(\xi_{x0}\cdot x-\frac{%
1}{2}\xi_{x0} \cdot x_0)}\Psi(x-x_0,p)
\end{equation*}
belongs to $\mathcal{S}(\mathbb{R}^{2n})$ for all $(x,p)$.

Now consider the following

\begin{definition}
Let the \textit{phase space Heisenberg-Weyl operator} be defined by
\begin{equation*}
\hat{T}_{PS}(z_{0}):L^{2}(\mathbb{R}^{2n})\longrightarrow L^{2}(\mathbb{R}%
^{2n})
\end{equation*}%
\begin{equation}
\hat{T}_{PS}(z_{0})\Psi (x,p):=e^{i(\xi _{x0}\cdot x-\frac{1}{2}\xi
_{x0}\cdot x_{0})}\Psi (x-x_{0},p).  \label{IVPSHW}
\end{equation}
\end{definition}

Notice that

\begin{remark}
From Remark 4.1 it is trivial to conclude that
\begin{equation*}
\hat{T}_{PS}(z_{0})=\left. \hat{U}_{PS}(z_{0},t)\right\vert _{t=1}
\end{equation*}%
where ($\hat{Z}_{x}=(\hat{X},\hat{\Xi}_{x})$, cf. Remark 2.8)
\begin{equation*}
\hat{U}_{PS}(z_{0},t)=e^{-it\sigma (\hat{Z}_{x},z_{0})}
\end{equation*}%
is the one-parameter unitary evolution group with infinitesimal generator $%
\sigma (\hat{Z}_{x},z_{0})$. Moreover, from the definition of $\hat{T}%
_{PS}(z_{0})$ we also realize that
\begin{equation*}
\left. \hat{T}_{PS}(z_{0})\right\vert _{\mathcal{S}_{\chi }}=\left. T_{\chi }%
\hat{T}(z_{0})T_{\chi }^{\ast }\right\vert _{\mathcal{S}_{\chi }}.
\end{equation*}
\end{remark}

It follows that $\hat{A}_{\chi }^{W}$ can be written as
\begin{equation}
\hat{A}_{\chi }^{W}=\left( \frac{1}{2\pi }\right) ^{n}\int_{\mathbb{R}^{2n}}%
\mathcal{F}_{\sigma }a(z_{0})\hat{T}_{PS}(z_{0})\,dz_{0}.  \label{IVAW1}
\end{equation}%
Now note that the functional form of the previous operator is independent of
$\chi $ and that its action can be consistently extended to $\mathcal{S}(%
\mathbb{R}^{2n})$. This naturally suggests

\begin{definition}
Let the \textit{phase space Weyl operator} $\hat A^W:\mathcal{S}(\mathbb{R}%
^{2n}) \to \mathcal{S}^{\prime }(\mathbb{R}^{2n})$, associated with the Weyl
symbol $a(x, \xi_x)$, be defined by
\begin{equation}
\hat A^W=\left(\frac{1}{2\pi}\right)^n \int_{\mathbb{R}^{2n}} \mathcal{F}%
_{\sigma}a(z_0) \hat T_{PS}(z_0) \, dz_0  \label{IVAW2}
\end{equation}
and let us denote by $a \overset{\text{Weyl}}{\longleftrightarrow} \hat A^W$
(or, more precisely, $A \overset{\text{Weyl}}{\longleftrightarrow} \hat A^W$%
; see the next Theorem) the \textit{phase space Weyl correspondence} between
$a$ (or $A$) and $\hat A^W$.
\end{definition}

We now prove that $\hat A^W$ is indeed a Weyl operator, whose restriction to
$\mathcal{S}_\chi$ satisfies $\hat A^W |_{\mathcal{S}_\chi}=\hat A^W_\chi$
for all $\chi \in \mathcal{S}(\mathbb{R}^n)$.

\begin{theorem}
Let $a(x, \xi_x) \in \mathcal{S}^{\prime }(\mathbb{R}^{2n})$ be the Weyl
symbol of $\hat a^W$. Then the operator $\hat A^W: \mathcal{S}(\mathbb{R}%
^{2n}) \longrightarrow \mathcal{S}^{\prime }(\mathbb{R}^{2n})$ defined by
\begin{equation}
\hat A^W \Psi=\left(\frac{1}{2\pi}\right)^n \langle \mathcal{F}%
_{\sigma}a(\cdot) ,\hat T_{PS}(\cdot)\Psi\rangle  \label{IVAW3}
\end{equation}
that is, formally, by eq.(\ref{IVAW2}), is a Weyl operator with symbol $A=a
\otimes 1$, in coordinates
\begin{equation}
A(x,p,\xi_x,\xi_p)=a(x,\xi_x)  \label{IVAa}
\end{equation}
and so $A \in \mathcal{S}^{\prime }(\mathbb{R}^{2n}\oplus \mathbb{R}^{2n})$.
\end{theorem}

\begin{proof}
Let $\Psi \in \mathcal{S}(\mathbb{R}^{2n})$. Since $\hat{T}_{PS}\Psi (z)\in
\mathcal{S}(\mathbb{R}^{2n})$ for all $z$ and $\mathcal{F}_{\sigma }a\in
\mathcal{S}^{\prime }(\mathbb{R}^{2n})$ the operator $\hat{A}^{W}$, given by
eq.(\ref{IVAW3}), is well-defined. We then have
\begin{eqnarray*}
\hat{A}^{W}\Psi (x,p) &=&\left( \frac{1}{2\pi }\right) ^{n}\int_{\mathbb{R}%
^{n}\times \mathbb{R}^{n}}\mathcal{F}_{\sigma }a(x_{0},\xi _{x0})e^{i(\xi
_{x0}\cdot x-\frac{1}{2}\xi _{x0}\cdot x_{0})}\Psi (x-x_{0},p)\,dx_{0}d\xi
_{x0} \\
&=&\left( \frac{1}{2\pi }\right) ^{n}\int_{\mathbb{R}^{n}\times \mathbb{R}%
^{n}}\mathcal{F}_{\sigma }a(x-x^{\prime },\xi _{x0})e^{\frac{i}{2}(\xi
_{x0}\cdot x+\xi _{x0}\cdot x^{\prime })}\Psi (x^{\prime },p)\,dx^{\prime
}d\xi _{x0}
\end{eqnarray*}%
where we performed the substitution $x^{\prime }=x-x_{0}$. Hence, the action
of $\hat{A}^{W}$ can be written
\begin{equation*}
\hat{A}^{W}\Psi (x,p)=\langle K_{A}(x,p;x^{\prime },p^{\prime }),\Psi
(x^{\prime },p^{\prime })\rangle
\end{equation*}%
for the kernel $K_{A}\in \mathcal{S}^{\prime }(\mathbb{R}^{2n}\times \mathbb{%
R}^{2n})$ given by
\begin{equation*}
K_{A}(x,p;x^{\prime },p^{\prime })=\left( \frac{1}{2\pi }\right) ^{n}\delta
(p-p^{\prime })\int_{\mathbb{R}^{n}}\mathcal{F}_{\sigma }a(x-x^{\prime },\xi
_{x0})e^{\frac{i}{2}(\xi _{x0}\cdot x+\xi _{x0}\cdot x^{\prime })}\,d\xi
_{x0}
\end{equation*}%
where the integral is interpreted in the distributional sense. Comparing
this expression with eq.(\ref{IV4}) we find that
\begin{equation*}
K_{A}(x,p;x^{\prime },p^{\prime })=K_{a}(x;x^{\prime })\delta (p-p^{\prime })
\end{equation*}%
where $K_{a}\in \mathcal{S}^{\prime }(\mathbb{R}^{n}\times \mathbb{R}^{n})$
is the kernel of the Weyl operator $\hat{a}^{W}$.

From eq.(\ref{IV2}) it follows that the Weyl symbol of $\hat A^W$ is
\begin{equation*}
A(x,p,\xi_x,\xi_p)=\int_{\mathbb{R}^{n}\times \mathbb{R}^n} K_A(x+\frac{1}{2}%
\eta_x,p+\frac{1}{2}\eta_p;x-\frac{1}{2}\eta_x,p-\frac{1}{2}\eta_p)
e^{-i(\xi_x \cdot \eta_x+\xi_p \cdot \eta_p)} \, d\eta_x d\eta_p
\end{equation*}
\begin{equation*}
=\int_{\mathbb{R}^{n}\times \mathbb{R}^n} K_a(x+\frac{1}{2}\eta_x;x-\frac{1}{%
2}\eta_x) \delta(\eta_p) e^{-i(\xi_x \cdot \eta_x+\xi_p \cdot \eta_p)} \,
d\eta_x d\eta_p
\end{equation*}
\begin{equation*}
=\int_{\mathbb{R}^{n}} K_a(x+\frac{1}{2}\eta_x;x-\frac{1}{2}\eta_x)
e^{-i\xi_x \cdot \eta_x} \, d\eta_x =a(x,\xi_x)
\end{equation*}
\end{proof}

\begin{theorem}
The operator $\hat{A}^{W}$ satisfies
\begin{equation}
\left. \hat{A}^{W}\right\vert _{\mathcal{S}_{\chi }}=\hat{A}_{\chi }^{W}
\label{IVRS}
\end{equation}%
and we have the intertwining relations
\begin{equation}
\hat{A}^{W}T_{\chi }=T_{\chi }\hat{a}^{W}  \label{IVIn1}
\end{equation}%
and
\begin{equation}
T_{\chi }^{\ast }\hat{A}^{W}=\hat{a}^{W}T_{\chi }^{\ast }  \label{IVIn2}
\end{equation}
\end{theorem}

\begin{proof}
The identity $\left. \hat A^W \right|_{\mathcal{S}_\chi}=\hat A^W_\chi$ is a
direct consequence of eqs.(\ref{IVAW1},\ref{IVAW2}) and $\mathcal{S}_\chi
\subset \mathcal{S}(\mathbb{R}^{2n})$.

Let us prove the first intertwining relation. For every $\psi \in D(\hat
a^W)=\mathcal{S}(\mathbb{R}^n)$ we have $T_\chi[\psi] \in \mathcal{S}_\chi$
and so
\begin{equation*}
\hat A^W T_\chi \psi = \hat A^W_\chi T_\chi \psi= T_\chi \hat a^W T_\chi^*
T_\chi \psi=T_\chi \hat a^W \psi
\end{equation*}
where we used the identity $T_\chi^* T_\chi=1$ in $\mathcal{S}(\mathbb{R}^n)$%
.

The second intertwining relation follows from
\begin{equation}
T_{\chi }^{\ast }\hat{T}_{PS}(z_{0})=\hat{T}(z_{0})T_{\chi }^{\ast }
\label{IVInT}
\end{equation}%
by interchanging the (Bochner) integrals (eqs.(\ref{IVaw}) and (\ref{IVAW2}%
)) with the operator $T_{\chi }^{\ast }$. The proof of eq.(\ref{IVInT}) in $%
\mathcal{S}(\mathbb{R}^{2n})$ is straightforward
\begin{eqnarray*}
T_{\chi }^{\ast }\hat{T}_{PS}(z_{0})\Psi (x,p) &=&T_{\chi }^{\ast }e^{i(\xi
_{x0}\cdot x-\frac{1}{2}\xi _{x0}\cdot x_{0})}\Psi (x-x_{0},p) \\
&=&\int \,e^{i(\xi _{x0}\cdot x-\frac{1}{2}\xi _{x0}\cdot x_{0})}\Psi
(x-x_{0},p)\chi (p)\,dp \\
&=&\hat{T}(z_{0})\int \,\Psi (x,p)\chi (p)\,dp=\hat{T}(z_{0})T_{\chi }^{\ast
}\Psi (x,p).
\end{eqnarray*}
\end{proof}

The spectral results for the operators $\hat A^W$ follow from the ones for $%
\hat a^W$ by a direct application of the previous Theorem.

\begin{corollary}
Let $\hat a^W$ and $\hat A^W$ be the Weyl and phase space Weyl operators
associated with the symbol $a\in \mathcal{S}(\mathbb{R}^{2n})$,
respectively. Then

(i) The eigenvalues of $\hat a^W$ and $\hat A^W$ are the same.

(ii) Let $\chi \in \mathcal{S}(\mathbb{R}^n)$. If $\psi_\lambda$ is an
eigenfunction of $\hat a^W$ then $\Psi_\lambda=T_\chi[\psi_\lambda]$ is an
eigenfunction of $\hat A^W$ (associated with the same eigenvalue).

(iii) Conversely, let $\Psi_\lambda$ be an eigenfunction of $\hat A^W$. If $%
\psi_\lambda=T_\chi^*[\Psi_\lambda]\not=0$ then $\psi_\lambda$ is an
eigenfunction of $\hat a^W$ (associated with the same eigenvalue).
\end{corollary}

\begin{proof}
Let $\chi \in \mathcal{S}(\mathbb{R}^{n})$ and let $\psi _{\lambda }$ be an
eigenfunction of $\hat{a}^{W}$
\begin{equation*}
\hat{a}^{W}\psi _{\lambda }=\lambda \psi _{\lambda }.
\end{equation*}%
Then $\Psi _{\lambda }=T_{\chi }[\psi _{\lambda }]$ is an eigenfunction of $%
\hat{A}_{\chi }^{W}$ (associated with the same eigenvalue)
\begin{equation*}
\hat{A}^{W}\Psi _{\lambda }=\hat{A}^{W}T_{\chi }[\psi _{\lambda }]=T_{\chi }%
\hat{a}^{W}\psi _{\lambda }=\lambda T_{\chi }[\psi _{\lambda }]=\lambda \Psi
_{\lambda }
\end{equation*}%
where we used eq.(\ref{IVIn1}). This proves (ii).

Conversely, assume that
\begin{equation*}
\hat{A}^{W}\Psi _{\lambda }=\lambda \Psi _{\lambda }
\end{equation*}%
and let $\chi \in \mathcal{S}(\mathbb{R}^{n})$ be such that $\psi _{\lambda
}=T_{\chi }^{\ast }[\Psi _{\lambda }]\not=0$. Then
\begin{equation*}
\hat{a}^{W}\psi _{\lambda }=\hat{a}^{W}T_{\chi }^{\ast }\Psi _{\lambda
}=T_{\chi }^{\ast }\hat{A}^{W}\Psi _{\lambda }=\lambda T_{\chi }^{\ast }\Psi
_{\lambda }=\lambda \psi _{\lambda }
\end{equation*}%
where we used eq.(\ref{IVIn2}). Hence $\psi _{\lambda }$ is also an
eigenfunction of $\hat{a}^{W}$ (associated with the same eigenvalue) which
proves (iii).

To conclude the proof of (i) we just notice that for every $\Psi_\lambda \in
\mathcal{S}(\mathbb{R}^{2n})-\{0\}$ there is always some $\chi \in \mathcal{S%
}(\mathbb{R}^n)$ such that $T_\chi^*[\Psi_\lambda] \not=0$.
\end{proof}

\section{Moyal representation}

In this section we construct another phase space representation of quantum
mechanics, which is intimately connected with the deformation quantization
of Bayen et al \cite{BFFLS1,BFFLS2}. Namely, the eigenvalue equation (for a
generic Weyl operator in this representation) is just the Moyal $\star $%
-genvalue equation (for the Weyl symbol of that operator) and its solutions
are thus the $\star $-genfunctions (Theorem 5.8 and Corollary 5.9).
Moreover, the Schr\"{o}dinger equation (in this representation) can be
written in terms of the Moyal starproduct and the Weyl symbol of the
Hamiltonian operator (Corollary 5.10). For these reasons we shall call it
the \textit{Moyal representation}. This formulation of quantum mechanics has
been studied before (also for the more general case where the canonical
structure is given by the extended Heisenberg algebra \cite{DGLP2}) using a
set of partial isometries $L^{2}(\mathbb{R}^{n})\rightarrow L^{2}(\mathbb{R}%
^{2n})$ (the windowed wavepacket transform, familiar from time-frequency
analysis \cite{gr00}) mapping the standard Schr\"{o}dinger configuration
space representation into the Moyal phase space representation \cite%
{CPDE,birkbis,GOLU1,GOLU2}. Here we will follow a different approach by
showing that the Schr\"{o}dinger phase space representation and the Moyal
representation are related by the unitary (in fact metaplectic)
transformation $U$ associated with the symplectic transformation eq.(\ref{I3}%
). This result allows us to translate the Weyl pseudo-differential calculus
and the spectral and dynamical results of the Schr\"{o}dinger phase space
representation directly into the Moyal representation.

In the next subsection we will determine the unitary transformation $U$
explicitly and show that its action on $\Psi \in \mathcal{S}_{\hat\chi}$ is
nothing else but the cross Wigner function associated with the density
matrix element $|\psi\rangle \langle \chi|$. In subsection 5.2 we use the
transformation $U$ to determine the Weyl pseudo-differential calculus in the
Moyal representation. Finally, in subsection 5.3 we construct the eigenvalue
and dynamical equations in this representation and study their relation with
the deformation quantization formulation.

\subsection{Unitary transformation}

In this section we determine the one-parameter quantum evolution groups
generated by the operators
\begin{eqnarray}
\hat H_D & = & \frac{1}{2} (\hat X \cdot \hat\Xi_x+\hat\Xi_x \cdot \hat X) +
\frac{1}{2} (\hat P\cdot \hat\Xi_p+\hat\Xi_p\cdot \hat P)  \notag \\
\hat H_R &=& -\hat\Xi_x\cdot\hat\Xi_p-\hat X\cdot \hat P  \label{VH}
\end{eqnarray}
which are obtained after quantizing the Hamiltonians (\ref{IHD},\ref{IHR})
and use then to determine the explicit form of the unitary operator $U$.
Note that the operators $\hat X,\hat P,\hat\Xi_x,\hat\Xi_p$ are given in the
phase space Schr\"odinger representation
\begin{equation}
\hat X=x=(x_1,...,x_n) \quad , \quad \hat P=p=(p_1,...,p_n)  \label{VSR}
\end{equation}
\begin{equation*}
\hat\Xi_x=-i\partial_x=(-i\partial_{x1},...,-i\partial_{xn}), \quad
\hat\Xi_p=-i\partial_p=(-i\partial_{p1},...,-i\partial_{pn})
\end{equation*}

\begin{remark}
Both $\hat{H}_{D}$ and $\hat{H}_{R}$ are self-adjoint on their maximal
domains
\begin{equation*}
D_{max}(\hat{H}_{D})=\{\Psi \in L^{2}(\mathbb{R}^{2n}):\hat{H}_{D}\Psi \in
L^{2}(\mathbb{R}^{2n})\}
\end{equation*}%
\begin{equation*}
D_{max}(\hat{H}_{R})=\{\Psi \in L^{2}(\mathbb{R}^{2n}):H_{R}\Psi \in L^{2}(%
\mathbb{R}^{2n})\}.
\end{equation*}%
This immediately follows from the fact that they are polynomial differential
expressions which are formally self-adjoint \cite{Gitman}
\end{remark}

We then have

\begin{theorem}
The one-parameter quantum evolution group generated by $\hat H_D$ is given
by ($s \in \mathbb{R}$)
\begin{equation}
U_D(s):L^2(\mathbb{R}^{2n}) \to L^2(\mathbb{R}^{2n}); \quad U_D(s) \Psi(x,p)
= e^{-ns} \Psi(e^{-s}x,e^{-s}p)  \label{VUD}
\end{equation}
\end{theorem}

\begin{proof}
Since $\hat{H}_{D}$ is self adjoint and $\mathcal{S}(\mathbb{R}^{2n})\subset
D(\hat{H}_{D})$
\begin{equation*}
\Psi (x,p,s)=e^{-is\hat{H}_{D}}\Psi _{0}(x,p)
\end{equation*}%
is the unique solution of (cf. [\cite{Oliveira}, chapter 5])
\begin{equation*}
i\frac{\partial \Psi }{\partial s}=\hat{H}_{D}\Psi \quad ,\quad \Psi (\cdot
,0)=\Psi _{0}\in \mathcal{S}(\mathbb{R}^{2n}).
\end{equation*}%
This equation reads
\begin{equation*}
i\frac{\partial \Psi }{\partial s}=-i(x\cdot \partial _{x}+p\cdot \partial
_{p}+n)\Psi
\end{equation*}%
and it is trivial to check that
\begin{equation*}
\Psi (x,p,s)=e^{-ns}\Psi _{0}(e^{-s}x,e^{-s}p)
\end{equation*}%
is its explicit solution. Defining
\begin{equation*}
U_{D}(s):L^{2}(\mathbb{R}^{2n})\rightarrow L^{2}(\mathbb{R}^{2n});\quad
U_{D}(s)\Psi (x,p):=e^{-ns}\Psi (e^{-s}x,e^{-s}p)
\end{equation*}%
one can also easily check that $U_{D}(s)$ is linear and unitary (so bounded
and continuous) operator. Since
\begin{equation*}
\left. U_{D}(s)\right\vert _{\mathcal{S}(\mathbb{R}^{2n})}=\left. e^{-is\hat{%
H}_{D}}\right\vert _{\mathcal{S}(\mathbb{R}^{2n})}
\end{equation*}%
and both operators are continuous (and $\mathcal{S}(\mathbb{R}^{2n})$ is
dense in $L^{2}(\mathbb{R}^{2n})$),
\begin{equation*}
U_{D}(s)=e^{-is\hat{H}_{D}}.
\end{equation*}
\end{proof}

\begin{theorem}
The one-parameter unitary evolution group generated by $\hat H_R$ is
\begin{equation*}
U_R(\theta): L^2(\mathbb{R}^{2n}) \longrightarrow L^2(\mathbb{R}^{2n})
\end{equation*}
\begin{equation}
U_R(\theta) \Psi(x,p) :=\mathcal{F}_p^{-1} \left[\hat\Psi(x\cos
\theta+\xi_p\sin \theta, \xi_p \cos \theta -x\sin\theta) \right]  \label{VUR}
\end{equation}
where
\begin{equation*}
\hat \Psi(x,\xi_p)=\mathcal{F}_p[\Psi(x,p)](x,\xi_p)= \frac{1}{(2\pi)^{n/2}}%
\int_{\mathbb{R}^n} \, e^{-i\xi_p\cdot p}\Psi(x,p) \, dp
\end{equation*}
and
\begin{equation*}
\mathcal{F}_p^{-1}[\Psi(x,\xi_p)](x,p)= \frac{1}{(2\pi)^{n/2}}\int_{\mathbb{R%
}^n} \, e^{i\xi_p\cdot p}\Psi(x,\xi_p) \, d\xi_p
\end{equation*}
are the partial and inverse partial Fourier transforms, defined as unitary
operators on $L^2(\mathbb{R}^{2n})$.
\end{theorem}

\begin{proof}
Since $\hat{H}_{R}$ is self-adjoint and $\mathcal{S}(\mathbb{R}^{2n})\subset
D(\hat{H}_{R})$, the unique solution of the initial value problem
\begin{equation}
i\frac{\partial \Psi }{\partial \theta }=\hat{H}_{R}\Psi \quad ,\quad \Psi
(\cdot ,0)=\Psi _{0}\in \mathcal{S}(\mathbb{R}^{2n})  \label{VHR1}
\end{equation}%
is given by
\begin{equation}
\Psi (x,p,\theta )=e^{-i\theta \hat{H}_{R}}\Psi _{0}(x,p)  \label{VHR2}
\end{equation}%
Since
\begin{equation*}
\hat{H}_{R}=\partial _{x}\cdot \partial _{p}-x\cdot p=\mathcal{F}_{p}^{-1}%
\left[ i\xi _{p}\cdot \partial _{x}-ix\cdot \partial _{\xi _{p}}\right]
\mathcal{F}_{p}
\end{equation*}
we conclude, defining $\hat{\Psi}(x,\xi _{p},\theta )=\mathcal{F}_{p}[\Psi
(x,p,\theta )]$, that $\hat{\Psi}$ satisfies the initial value problem
\begin{equation*}
i\frac{\partial \hat{\Psi}}{\partial \theta }=\left[ i\xi _{p}\cdot \partial
_{x}-ix\cdot \partial _{\xi _{p}}\right] \hat{\Psi},\quad \hat{\Psi}(\cdot
,0)=\hat{\Psi}_{0}\in \mathcal{S}(\mathbb{R}^{2n})
\end{equation*}%
whose solution
\begin{equation*}
\hat{\Psi}(x,\xi _{p},\theta )=\hat{\Psi}_{0}(x(\theta ),\xi _{p}(\theta ))
\end{equation*}%
\begin{equation*}
x(\theta )=x\cos \theta +\xi _{p}\sin \theta \quad ,\quad \xi _{p}(\theta
)=-x\sin \theta +\xi _{p}\cos \theta
\end{equation*}
always is in $\mathcal{S}(\mathbb{R}^{2n})$.

Hence, the unique solution of (\ref{VHR1}) is
\begin{equation*}
\Psi (x,p,\theta )=\mathcal{F}_{p}^{-1}\left[ \hat{\Psi}(x,\xi _{p},\theta )%
\right] =\mathcal{F}_{p}^{-1}\left[ \hat{\Psi}_{0}(x\cos \theta +\xi
_{p}\sin \theta ,-x\sin \theta +\xi _{p}\cos \theta )\right]
\end{equation*}%
Defining
\begin{equation*}
U_{R}(\theta ):L^{2}(\mathbb{R}^{2n})\rightarrow L^{2}(\mathbb{R}^{2n})
\end{equation*}%
\begin{equation*}
U_{R}(\theta )\Psi (x,p):=\mathcal{F}_{p}^{-1}\left[ \hat{\Psi}(x\cos \theta
+\xi _{p}\sin \theta ,-x\sin \theta +\xi _{p}\cos \theta )\right]
\end{equation*}%
it is trivial to check that $U_{R}(\theta )$ is linear, unitary (so bounded
and continuous) operator. Since by (\ref{VHR2})
\begin{equation*}
\left. U_{R}(\theta )\right\vert _{\mathcal{S}(\mathbb{R}^{2n})}=\left.
e^{-i\theta \hat{H}_{R}}\right\vert _{\mathcal{S}(\mathbb{R}^{2n})}
\end{equation*}%
and both operators are continuous (and $\mathcal{S}(\mathbb{R}^{2n})$ is
dense in $L^{2}(\mathbb{R}^{2n})$),
\begin{equation*}
U_{R}(\theta )=e^{-i\theta \hat{H}_{R}}.
\end{equation*}
\end{proof}

We are now in position to define the unitary operator that corresponds to
the symplectic transformation eq.(\ref{I3})

\begin{corollary}
Consider the unitary operator
\begin{equation*}
U: L^2(\mathbb{R}^{2n}) \longrightarrow L^2(\mathbb{R}^{2n})
\end{equation*}
defined by
\begin{equation}
U:=U_D^{-1}(\ln \sqrt{2})U_R^{-1}(\pi /4)  \label{VU-def}
\end{equation}
Then $U$ acts as
\begin{equation}
U \Psi(x,p)=\mathcal{F}^{-1}_p \left[\hat\Psi(x-\xi_p/2,x+\xi_p/2)\right]%
(x,p)  \label{VU}
\end{equation}
and for
\begin{equation*}
\Psi(x,p)=T_{\hat \chi}[\psi](x,p)=\psi(x) \hat\chi^*(p)
\end{equation*}
where $\hat\chi(p)=\mathcal{F}_{\xi_p}[\chi(\xi_p)](p)$ and $\hat\chi^*(p)=%
\overline{\hat\chi(p)}$, we have
\begin{equation}
U\Psi(x,p)=\frac{1}{(2 \pi)^{n/2}} \int \, e^{ip\cdot\xi_p} \psi(x-\frac{%
\xi_p}{2}) \chi^*(x+\frac{\xi_p}{2})\, d\xi_p  \label{V-WF}
\end{equation}
which is (up to a factor of $(2\pi)^{n/2}$) the cross Wigner function
associated with the density matrix element $|\psi\rangle \langle \chi|$,
i.e. $UT_{\hat\chi}[\psi]=W(\psi,\chi)$.
\end{corollary}

\begin{proof}
We have
\begin{equation*}
U_{R}^{-1}(\pi /4)\Psi (x,p)=U_{R}(-\pi /4)\Psi (x,p)=\mathcal{F}_{p}^{-1}%
\left[ \hat{\Psi}(\frac{\sqrt{2}}{2}(x-\xi _{p}),\frac{\sqrt{2}}{2}(x+\xi
_{p}))\right]
\end{equation*}%
\begin{equation*}
=\frac{1}{(2\pi )^{n/2}}\int_{\mathbb{R}^{n}}\,e^{i\xi _{p}\cdot p}\hat{\Psi}%
(\frac{\sqrt{2}}{2}(x-\xi _{p}),\frac{\sqrt{2}}{2}(x+\xi _{p}))\,d\xi _{p}
\end{equation*}%
and since for arbitrary $\Phi (x,p)\in L^{2}(\mathbb{R}^{2n})$
\begin{eqnarray*}
U_{D}^{-1}(\ln \sqrt{2})\Phi (x,p) &=&U_{D}(-\ln \sqrt{2})\Phi (x,p) \\
&=&e^{n\ln \sqrt{2}}\Phi (e^{\ln \sqrt{2}}x,e^{\ln \sqrt{2}}p) \\
&=&2^{n/2}\Phi (\sqrt{2}x,\sqrt{2}p)
\end{eqnarray*}%
we get
\begin{eqnarray*}
U\Psi (x,p) &=&U_{D}^{-1}(\ln \sqrt{2})\frac{1}{(2\pi )^{n/2}}\int_{\mathbb{R%
}^{n}}\,e^{i\xi _{p}\cdot p}\hat{\Psi}(\frac{\sqrt{2}}{2}(x-\xi _{p}),\frac{%
\sqrt{2}}{2}(x+\xi _{p}))\,d\xi _{p} \\
&=&\frac{2^{n/2}}{(2\pi )^{n/2}}\int_{\mathbb{R}^{n}}\,e^{i\sqrt{2}\xi
_{p}\cdot p}\hat{\Psi}(\frac{\sqrt{2}}{2}(\sqrt{2}x-\xi _{p}),\frac{\sqrt{2}%
}{2}(\sqrt{2}x+\xi _{p}))\,d\xi _{p} \\
&=&\frac{1}{(2\pi )^{n/2}}\int_{\mathbb{R}^{n}}\,e^{i\xi _{p}^{\prime }\cdot
p}\hat{\Psi}(x-\frac{\xi _{p}^{\prime }}{2},x+\frac{\xi _{p}^{\prime }}{2}%
)\,d\xi _{p}^{\prime } \\
&=&\mathcal{F}_{p}^{-1}\left[ \hat{\Psi}(x-\xi _{p}/2,x+\xi _{p}/2)\right]
(x,p)
\end{eqnarray*}%
where we made the change of variable $\xi _{p}^{\prime }=\sqrt{2}\xi _{p}$.
This proves eq.(\ref{VU}).

To prove eq.(\ref{V-WF}) consider eq. (\ref{VU}) for $\Psi(x,p)=\psi(x) \hat
\chi^*(p)$. We have
\begin{equation*}
\hat \Psi(x,\xi_p)=\mathcal{F}_p[\psi(x)\hat\chi^*(p)](x,\xi_p)=\psi(x)
\mathcal{F}_p[\hat \chi^*(p)](\xi_p)
\end{equation*}
since
\begin{equation*}
\hat\chi^*(p)=\overline{\mathcal{F}_{\xi_p}[\chi(\xi_p)](p)}=\mathcal{F}%
_p^{-1}[ \chi^*(\xi_p)](p)
\end{equation*}
we get
\begin{equation*}
\hat \Psi(x,\xi_p)=\mathcal{F}_p[\Psi(x,p)]=\psi(x)\chi^*(\xi_p)
\end{equation*}
and the result follows.
\end{proof}

Finally, we consider the action of the unitary transformation on the
fundamental operators

\begin{theorem}
The operator $U$ maps the Schr\"odinger phase space representation (\ref{VSR}%
) into the Moyal representation of the Heisenberg algebra on $L^2(\mathbb{R}%
^{2n})$
\begin{equation}
\begin{array}{ll}
\tilde X & =U\hat XU^{-1}=\hat X-\frac{\hat\Xi_p}{2}=x+i\frac{1}{2}\partial_p
\\
\tilde P & =U\hat P U^{-1}=\hat P-\frac{\hat\Xi_x}{2}=p+i\frac{1}{2}%
\partial_x \\
\tilde \Xi_x & =U\hat \Xi_x U^{-1}=\hat P+\frac{\hat\Xi_x}{2}=p-i\frac{1}{2}%
\partial_x \\
\tilde \Xi_p & =U\hat \Xi_p U^{-1}=\hat X+\frac{\hat\Xi_p}{2}=x-i\frac{1}{2}%
\partial_p%
\end{array}
\label{Vop}
\end{equation}
\end{theorem}

\begin{proof}
We first notice that the operators $\hat{H}_{D}$ and $\hat{H}_{R}$ are
self-adjoint and quadratic. Hence, the solution of the Heisenberg equations
of motion
\begin{equation*}
i\frac{\partial }{\partial t}\hat{Z}=[\hat{Z},\hat{H}]\quad ,\quad \hat{Z}=%
\hat{X},\,\hat{P},\,\hat{\Xi}_{X},\,\hat{\Xi}_{P}
\end{equation*}%
coincides with the classical solution for both $\hat{H}=\hat{H}_{D}$ and $%
\hat{H}=\hat{H}_{R}$. The operators (\ref{Vop}) are the $\hat{H}_{R}$%
-evolution (up to $t=\pi /4$) of the $\hat{H}_{D}$-evolution (up to $t=\ln
\sqrt{2}$) of $\hat{X},\hat{P},\hat{\Xi}_{X}$ and $\hat{\Xi}_{P}$. The
solution (\ref{Vop}) then follows from the equivalent classical solutions
that were obtained in section 1.1. We have, for instance, for $\tilde{X}$
\begin{equation*}
\tilde{X}=U\hat{X}U^{-1}=e^{i(\ln \sqrt{2})\hat{H}_{D}}e^{i\frac{\pi }{4}%
\hat{H}_{R}}\hat{X}e^{-i\frac{\pi }{4}\hat{H}_{R}}e^{-i(\ln \sqrt{2})\hat{H}%
_{D}}
\end{equation*}%
\begin{equation*}
=e^{i(\ln \sqrt{2})\hat{H}_{D}}\frac{\sqrt{2}}{2}\left( \hat{X}-\hat{\Xi}%
_{p}\right) e^{-i(\ln \sqrt{2})\hat{H}_{D}}=\frac{\sqrt{2}}{2}\left( \sqrt{2}%
\hat{X}-\frac{\sqrt{2}}{2}\hat{\Xi}_{p}\right) =\hat{X}-\frac{\hat{\Xi}_{p}}{%
2}.
\end{equation*}%
Finally, since the transformation is unitary, the commutation relations are
preserved and the operators (\ref{Vop}) provide a phase space representation
of the Heisenberg algebra.
\end{proof}

\subsection{Moyal-Weyl pseudo-differential calculus}

A generic operator in the Moyal representation can be obtained from the
corresponding operator in the phase space Schr\"odinger representation by
the action of the unitary transformation $U$. For the operators $\hat A^W$
the action of $U$ yields the \textit{Moyal-Weyl pseudo-differential operators%
}
\begin{equation}
\tilde A^W=U \hat A^W U^{-1} =\left(\frac{1}{2\pi}\right)^n \int_{\mathbb{R}%
^{2n}} F_{\sigma}a(z_0) U\hat T_{PS}(z_0)U^{-1} \, dz_0  \label{VAW1}
\end{equation}
where the domain of $U$ (\ref{VU}) was trivially extended to $\mathcal{S}%
^{\prime }(\mathbb{R}^{2n})$.

We then have

\begin{theorem}
The Moyal-Heisenberg-Weyl operator
\begin{equation*}
\tilde{T}_{M}(z_{0}):=U\hat{T}_{PS}(z_{0})U^{-1}:L^{2}(\mathbb{R}%
^{2n})\longrightarrow L^{2}(\mathbb{R}^{2n})
\end{equation*}%
is given explicitly by ($z_{0}=(x_{0},\xi _{x0})\in \mathbb{R}^{2n}$)
\begin{equation}
\tilde{T}_{M}(x_{0},\xi _{x0})\Psi (x,p)=e^{-i(x_{0}\cdot p-\xi _{x0}\cdot
x)}\Psi (x-\frac{x_{0}}{2},p-\frac{\xi _{x0}}{2}).  \label{VMHW}
\end{equation}
\end{theorem}

\begin{proof}
We recall from Remark 4.4 that $\hat{T}_{PS}(z_{0})$ is the unitary operator
\begin{equation*}
\hat{T}_{PS}(z_{0})=\left. e^{-it\hat{H}_{z_{0}}}\right\vert _{t=1}
\end{equation*}%
associated with the infinitesimal generator
\begin{equation*}
\hat{H}_{z_{0}}=\sigma (\hat{Z}_{x},z_{0})=x_{0}\cdot \hat{\Xi}_{x}-\xi
_{x0}\cdot \hat{X}
\end{equation*}%
Let now
\begin{equation*}
\tilde{T}(z_{0},t)=Ue^{-it\hat{H}_{z_{0}}}U^{-1}
\end{equation*}%
Then $\tilde{T}_{M}(z_{0})=\tilde{T}(z_{0},1)$ and $\Psi (z,t)=\tilde{T}%
(z_{0},t)\Psi _{0}(z)$ is the unique solution of the initial value problem
\begin{equation}
i\frac{\partial }{\partial t}\Psi =U\hat{H}_{z_{0}}U^{-1}\Psi ,\quad \Psi
(\cdot ,0)=\Psi _{0}\in \mathcal{S}(\mathbb{R}^{2n})  \label{VSM}
\end{equation}%
Since (cf. eq.(\ref{Vop}))
\begin{equation*}
U\hat{H}_{z_{0}}U^{-1}=x_{0}\cdot \hat{P}-\xi _{x0}\cdot \hat{X}+\tfrac{1}{2}%
(x_{0}\cdot \hat{\Xi}_{x}+\xi _{x0}\cdot \hat{\Xi}_{p})
\end{equation*}%
it is trivial to check that
\begin{equation*}
\Psi (z,t)=e^{-i(x_{0}\cdot p-\xi _{x0}\cdot x)t}\Psi _{0}(z-\tfrac{1}{2}%
z_{0}t)
\end{equation*}%
is a solution of eq.(\ref{VSM}). It follows that $\tilde{T}(z_{0},t)$ is
unitary and extends trivially to $L^{2}(\mathbb{R}^{2n})$. Hence, $\tilde{T}%
_{M}(z_{0})=\tilde{T}(z_{0},1)$ is, in fact, the unitary transformation
given by eq.(\ref{VMHW}).
\end{proof}

\begin{theorem}
Let $a(z_x) \in \mathcal{S}(\mathbb{R}^{2n})$ be the Weyl symbol of $\hat
a^W $. The operator $\tilde A^W:\mathcal{S}(\mathbb{R}^{2n})\to \mathcal{S}%
^{\prime }(\mathbb{R}^{2n})$ given by (\ref{VAW1}), explicitly
\begin{equation}
\tilde A^W\Psi(z)=\frac{1}{(2\pi)^n} \, \langle F_\sigma a(\cdot),\tilde
T_M(\cdot) \Psi(z)\rangle  \label{VMWO}
\end{equation}
is a Weyl operator with symbol
\begin{equation*}
A_M(x,p,\xi_x,\xi_p) =a(x-\frac{\xi_p}{2},p+\frac{\xi_x}{2})
\end{equation*}
and we shall write $a \overset{\text{Moyal-Weyl}}{\longleftrightarrow}
\tilde A^W$ for the "Moyal-Weyl correspondence" between $a$ and $\tilde A^W$.
\end{theorem}

\begin{proof}
A simple proof follows from Theorem 4.2 and the fact that $U^{-1}\in
\limfunc{Mp}(4n,\sigma _{x}\oplus \sigma _{p})$ (because it is a unitary
transformation generated by quadratic Hamiltonians) and projects into the
symplectic transformation $S\in Sp(4n;\sigma _{x}\oplus \sigma _{p})$
\begin{equation*}
S(x,p,\xi _{x},\xi _{p})=(x-\frac{\xi _{p}}{2},p-\frac{\xi _{x}}{2},p+\frac{%
\xi _{x}}{2},x+\frac{\xi _{p}}{2})
\end{equation*}%
that was calculated explicitly in section 1.1. Hence, a direct application
of Theorem 4.2 leads to the conclusion that
\begin{equation*}
\tilde{A}^{W}=U\hat{A}^{W}U^{-1}
\end{equation*}%
is the Weyl operator with symbol $A_{M}=A\circ S$ (where $A\overset{\text{%
Weyl}}{\longleftrightarrow }\hat{A}^{W}$), explicitly
\begin{equation*}
A_{M}(x,p,\xi _{x},\xi _{p})=A(x-\frac{\xi _{p}}{2},p-\frac{\xi _{x}}{2},p+%
\frac{\xi _{x}}{2},x+\frac{\xi _{p}}{2})
\end{equation*}%
Since $A$ satisfies (\ref{IVAa}) we also get
\begin{equation*}
A_{M}(x,p,\xi _{x},\xi _{p})=a(x-\frac{\xi _{p}}{2},p+\frac{\xi _{x}}{2})
\end{equation*}%
which concludes the proof.
\end{proof}

\subsection{Deformation Quantization}

In this section we succinctly discuss the relation between the Moyal
representation and the deformation quantization of Bayen \textit{et al.}
\cite{BFFLS1,BFFLS2}. For a complete presentation the reader should refer to
\cite{GOLU1,GOLU2,CPDE}

\begin{theorem}
Let $\tilde A^W:\mathcal{S}(\mathbb{R}^{2n}) \to \mathcal{S}^{\prime }(%
\mathbb{R}^{2n})$ be the Moyal-Weyl operator (\ref{VMWO}) written in terms
of the Weyl symbol $a\in \mathcal{S}^{\prime }(\mathbb{R}^{2n})$ of $\hat
a^W $, i.e. $a \overset{\text{Moyal-Weyl}}{\longleftrightarrow} \tilde A^W$.
Then, for all $\Psi \in \mathcal{S}(\mathbb{R}^{2n})$, we have
\begin{equation*}
\tilde A^W \Psi=a \star \Psi
\end{equation*}
where $*$ is the Moyal starproduct (\ref{IVMSP}).
\end{theorem}

\begin{proof}
For completeness we review the proof of \cite{GOLU1}.

For $\Psi \in \mathcal{S}(\mathbb{R}^{2n})$ we have
\begin{equation*}
\tilde A^W \Psi (z)=\left(\frac{1}{2\pi} \right)^n \int_{\mathbb{R}^{2n}}%
\mathcal{F}_{\sigma}a(z_0) \tilde T_M(z_0) \Psi(z)\, dz_0
\end{equation*}
\begin{equation*}
=\left(\frac{1}{2\pi} \right)^{2n} \int_{\mathbb{R}^{2n}} \left[ \int_{%
\mathbb{R}^{2n}}e^{-i \sigma(z_0,z_1)} a(z_1) \, dz_1 \right] e^{-i
\sigma(z,z_0)} \Psi(z-\frac{1}{2}z_0) \, dz_0
\end{equation*}
Letting $z_0=v$ and $z_1=z+\frac{1}{2}u$ and noticing that
\begin{equation*}
\sigma(z_0,z+\frac{1}{2}u)+\sigma(z,z_0)=\frac{1}{2}\sigma(z_0,u)
\end{equation*}
we find
\begin{equation*}
\tilde A^W \Psi(z)=\left(\frac{1}{4\pi} \right)^{2n} \int_{\mathbb{R}^{2n}}
\int_{\mathbb{R}^{2n}}e^{-\frac{i}{2} \sigma(v,u)}a(z+\frac{1}{2}u)\Psi( z-%
\frac{1}{2}v)\, du dv
\end{equation*}
which is exactly $a \star \Psi$ cf.(\ref{IVMSP}).
\end{proof}

Two immediate corollaries are

\begin{corollary}
$\Psi _{\lambda }$ is the right-stargenfunction of $a$, i.e.
\begin{equation*}
a\star \Psi _{\lambda }=\lambda \Psi
\end{equation*}%
iff $\Psi _{\lambda }$ is an eigenfunction of $\tilde{A}^{W}\overset{\text{%
Moyal-Weyl}}{\longleftrightarrow }a$, i.e.
\begin{equation*}
\tilde{A}^{W}\Psi _{\lambda }=\lambda \Psi _{\lambda }.
\end{equation*}
\end{corollary}

and:

\begin{corollary}
The Schr\"odinger equation in the Moyal representation can be written in the
form
\begin{equation*}
i \frac{\partial \Psi}{\partial t}=h\star\Psi
\end{equation*}
where $h \overset{\text{Moyal-Weyl}}{\longleftrightarrow} \tilde H^W$ and $%
\tilde H^W$ is the Hamiltonian of the system in the Moyal representation.
\end{corollary}

In view of the spectral relation of $\hat a^W$ and $\hat A^W$ (Corollary
4.5) and the unitary relation between $a\star=\tilde A^W=U \hat A^W U^{-1}$
and $\hat A^W$ (cf. eq.(\ref{VAW1})), we also have

\begin{corollary}
Let $a \in \mathcal{S}^{\prime }(\mathbb{R}^{2n})$ be the Weyl symbol of $%
\hat a^W$. Then

i) The eigenvalues of the stargenvalue equations
\begin{equation}
a\star \Psi _{\lambda }=\lambda \Psi _{\lambda }  \label{VSE}
\end{equation}%
and
\begin{equation}
\hat{a}^{W}\psi _{\lambda }=\lambda \psi _{\lambda }  \label{VEE}
\end{equation}%
are the same.

ii) If $\psi_\lambda$ is an eigenfunction (solution of eq.(\ref{VEE})) then $%
\Psi_\lambda=U T_\chi \psi_\lambda$ is a stargenfunction (solution of (\ref%
{VSE})) associated with the same eigenvalue.

iii) Conversely, if $\Psi_\lambda$ is a stargenfunction and $\psi_\lambda=
T_\chi^* U^{-1} \Psi_\lambda\not=0$ then $\psi_\lambda$ is an eigenfunction
of $\hat a^W$ (associated with the same eigenvalue).
\end{corollary}

\begin{proof}
The result follows from Corollary 4.5 by a direct application of the unitary
transformation $U$ (taking into account Corollary 5.10).
\end{proof}

\subsection*{Acknowledgements}

Maurice de Gosson has been financed by the Austrian Research Agency FWF
(Projektnummer P 23902-N13). Nuno Costa Dias and Jo\~{a}o Nuno Prata have
been supported by the grants PDTC/MAT/ 69635/2006 and PTDC/MAT/099880/2008
of the Portuguese Science Foundation (FCT). Franz Luef has been financed by
the Marie Curie Outgoing Fellowship PIOF 220464.

\textbf{Author's addresses:}\bigskip

\textbf{Franz Luef\footnote{%
franz.luef@univie.ac.at}:} \textit{Universit\"{a}t Wien, NuHAG}

\textit{\ Fakult\"{a}t f\"{u}r Mathematik }

\textit{Wien \ 1090, Austria}






\bigskip

\textbf{Jo\~{a}o Nuno Prata\footnote{%
joao.prata@mail.telepac.pt } and Nuno Costa Dias\footnote{%
ncdias@meo.pt}:}

\textit{Departamento de Matem\'{a}tica. Universidade Lus\'{o}fona de
Humanidades }

\textit{e Tecnologias. Av. Campo Grande, 376, }

\textit{1749-024 Lisboa, Portugal}

\textit{and}

\textit{Grupo de F\'{\i}sica Matem\'{a}tica, }

\textit{Universidade de Lisboa, }

\textit{Av. Prof. Gama Pinto 2, }

\textit{1649-003 Lisboa, Portugal}\bigskip

\textbf{Maurice de Gosson\footnote{%
maurice.de.gosson@univie.ac.at}:} \textit{Universit\"{a}t Wien, NuHAG}

\textit{\ Fakult\"{a}t f\"{u}r Mathematik }

\textit{Wien \ 1090, Austria}

\end{document}